\documentclass{ieeeaccess}

\usepackage{cite}
\usepackage{amsmath,amsthm,amsfonts, amssymb, array, subcaption, xcolor}
\usepackage[linesnumbered,ruled,vlined]{algorithm2e}
\usepackage{graphicx}
\usepackage{textcomp}
\def\BibTeX{{\rm B\kern-.05em{\sc i\kern-.025em b}\kern-.08em
    T\kern-.1667em\lower.7ex\hbox{E}\kern-.125emX}}
\usepackage[braket, qm]{qcircuit}
\usepackage{float}
\usepackage{url}

\newcommand{\bbF}{\mathbb F}

\newtheorem{thm}{Theorem}

\newtheorem{prop}[thm]{Proposition}

\theoremstyle{definition}

\newtheorem{rmk}[thm]{Remark}
\numberwithin{equation}{section}
\numberwithin{figure}{section}

\begin{document}
\history{Accepted February 8, 2022.}
\doi{10.1109/TQE.2020.DOI}

\title{A Grover Search-based Algorithm for the List Coloring Problem}
\author{\uppercase{Sayan Mukherjee}\authorrefmark{1}}
\address[1]{blueqat Co. Ltd., Tokyo 150-0002, Japan}
\address[1]{University of Illinois at Chicago, Chicago, IL-60608, USA}

\markboth
{Mukherjee: A Grover Search-based Algorithm for the List Coloring Problem}
{Mukherjee: A Grover Search-based Algorithm for the List Coloring Problem}

\corresp{(email: sayan@blueqat.com).}

\begin{abstract}
Graph coloring is a computationally difficult problem, and currently the best known classical algorithm for $k$-coloring of graphs on $n$ vertices has runtimes $\Omega(2^n)$ for $k\ge 5$. The list coloring problem asks the following more general question: given a \emph{list} of available colors for each vertex in a graph, does it admit a proper coloring? We propose a hybrid classical-quantum algorithm based on Grover search~\cite{grover-DatabaseSearch-1996} to quadratically speed up exhaustive search. Our algorithm loses in complexity to classical ones in specific restricted cases, but improves exhaustive search for cases where the lists and graphs considered are arbitrary in nature.
\end{abstract}


\titlepgskip=-15pt

\maketitle

\section{Introduction}
\label{sec:introduction}
\PARstart{G}{raph} coloring problems provide for a rich family of NP-complete problems in theoretical computer science. 
While exhaustive search is believed to be the fastest classical approach for several NP-complete problems including satisfiability and hitting-set~\cite{cygan-ProblemsCNFSAT-2016}, there are much better classical algorithms using dynamic programming, inclusion-exclusion and other structural approaches for problems such as graph coloring~\cite{eppstein-MISGraphCol-2001,bodlaender-PolyMemGraphCol-2006,koivisto-InclExclGraphCol-2006}, the traveling salesman problem~\cite{little-TSP-1963,hoffman-TSP-2013}, set cover~\cite{hua-DPSetCover-2010} etc.
Several authors have obtained quantum speedup on these classical algorithms~\cite{ambainis-DPQuantumSpeedup-2019,ronagh-DPQuantum-2019,shimizu-ExpTimeQuantAlgoGraphColoring-2021}; however, all of these algorithms have the limitation that they cannot be easily generalized to the list coloring problem.

Given a finite graph $G=(V,E)$, a proper coloring of $G$ is a function $\chi:V\to\mathbb N$ such that for every edge $uv\in E$, $\chi(u)\neq \chi(v)$. The list coloring problem tries to determine a proper coloring $\chi$ of a graph $G=(V,E)$, given a list $L_v$ of available colors for each vertex $v$. In other words, it is forced that $\chi(v)\in L_v$. When $L_v=\{1,2,\ldots,k\}$ for every vertex $v$ this reduces to the well-studied $k$-coloring problem.
We propose a simple Grover search-based approach to obtain a quadratic speedup on exhaustive search for the list coloring problem.

Grover's algorithm~\cite{grover-DatabaseSearch-1996} is known to speed up unstructured search quadratically using the technique of \textit{amplitude amplification}. In its simplest form, to find some marked elements from a list of $N=2^n$ entries, the algorithm starts with a uniform quantum superposition of all $2^n$ basis states of an $n$-qubit register. It then amplifies the amplitudes of the searched state and reduces those of the other states, such that a measurement of the $n$ qubits leads to one of the searched states with high probability.

Grover's algorithm has been used to obtain quantum speedups for various problems in combinatorial optimization and computer science (see, for e.g., \cite{kravchenko-SubtractionGames2019,Shukla-TrajectoryQuantumOpt-2019,Khadiev-StringProblems-2019,jeffery-BooleanMatrixMultGraphCollision-2016,lee-TriangleDetectionAssociativityTesting-2017}). Needless to say, graph coloring problems are also not an exception in the literature, and have been attacked using quantum annealing~\cite{titiloye-AnnealingColoring-2011,kudo-ConstrainedAnnealing-2018}, hybrid approaches~\cite{titiloye-HybridAnnealingColoring-2011,bravyi-HybridApproximateColoring-2020}, as well as using Grover search~\cite{wang-TernaryGrover-2011,shimizu-ExpTimeQuantAlgoGraphColoring-2021,saha-CircuitDesignColoringNearTerm-2020}.

In \cite{wang-TernaryGrover-2011}, a qutrit-based approach has been used to demonstrate the cost-efficiency of ternary quantum logic; however, their main algorithm is not realizable right now on NISQ devices. The algorithm of \cite{shimizu-ExpTimeQuantAlgoGraphColoring-2021} has the same issue as it requires quantum RAM which has not been realized at this moment.
On the other hand, the authors of \cite{saha-CircuitDesignColoringNearTerm-2020} and \cite{saha-SynthVertexColGrover-2015} demonstrate a quantum algorithm solving the $k$-coloring problem on NISQ devices, comparing the efficiency of their algorithm against the reduction of 3-SAT to 3-coloring approach of Hu et. al.~\cite{hu-ReductionBasedProblemMapping-2019}.

All of these algorithms use an oracle design which uses binary comparators, and provide solutions where almost all binary strings have positive probabilities of being selected, including those that do not represent valid colorings. Our approach circumvents this problem via a modified initialization and diffusion operator that restricts the evolution of the quantum algorithm to the only $\prod_{v\in L_v}|L_v|$ plausible states. Note that this is the total number of valid colorings when the underlying graph is empty. We achieve this via the restricted version of Grover search~\cite{grover-DatabaseSearch-1996,gilliam-GroverAdaptiveSearch-2021}.

\begin{prop}[Restricted Grover search]
\label{prop:subspaceGrover}
	Let $S\subseteq \{1,2,\ldots, 2^n-1\}$, and suppose $S'\subsetneq S$ is a set of marked states. Let $O$ be an oracle that marks these states and requires $a$ ancillas. Then, there is a quantum circuit on $n+a+1$ qubits which makes $O(\sqrt{2^n/|S'|})$ queries, which when measured, gives one of the marked states with high probability. Further, states outside $S$ are never measured.
\end{prop}

Additionally, we use an oracle design different from those in \cite{wang-TernaryGrover-2011,saha-CircuitDesignColoringNearTerm-2020}, and give a classical algorithm in Section 3 that can reduce the complexity of this oracle in several special cases (such as for the $3$-coloring or $4$-coloring problems). As a corollary of Proposition~\ref{prop:subspaceGrover}, our main theorem provides an algorithm for the list coloring problem.

\begin{thm}[Quantum list coloring algorithm]
\label{thm:qListColoring}
Given a graph $G=(V,E)$ on $n$ vertices and $m$ edges and lists of available colors $\{L_v: v\in V\}$, there exists a $(\sum_{v\in V}\lceil\log_2 |L_v|\rceil+m+1)$-qubit quantum algorithm with query complexity $O(\prod_{v\in V}|L_v|^{1/2})$ that returns a valid list coloring of $G$ with high probability.
\end{thm}

This paper is organized as follows. In Section~\ref{sec:groveralgo}, we describe Grover's algorithm and a gate-level implementation. Section~\ref{sec:qlistcoloring} is devoted to tackling the list coloring problem, and proves Theorem~\ref{thm:qListColoring}. In Section~\ref{sec:results}, we run experiments on classical simulators as well as real quantum machines, and compare the outcomes. We discuss applications and provide concluding remarks in Section~\ref{sec:conclusion}.

\section{Grover's Algorithm}
\label{sec:groveralgo}
In this section, we provide a concise exposition on Grover search.
The main idea behind Grover search is to amplify the amplitudes of some number of marked states (states which are being searched for), and consequentially decrease that of unmarked states.
Grover's Algorithm requires three different operators: Initialization, Oracle, and Diffusion. Below we present two formulations of the algorithm:

\subsection{Unrestricted search space}
When searching for a marked state among the full search space $S=\{0,1\}^n$, the initialization step of the algorithm creates a uniform superposition of all the possible states of an $n$-qubit system. This is achieved via appending Hadamard gates on each qubit: 
\[H^{\otimes n}\lvert 0\rangle_n = \lvert +\rangle^{n} = \frac{1}{\sqrt{2^n}}\sum_{i=1}^n \lvert i\rangle_n.\]
Here we abuse notation and write $\lvert i\rangle_n$ to denote the state corresponding to an $n$-digit binary representation of $i$.

Next, Grover's Algorithm requires an oracle $O$ that, given a uniform superposition of all $2^n$ possible states, can change the sign of the marked states. Let $S'\subseteq \{0,1,\ldots 2^n-1\}$ be a set of marked states. The Oracle $O$ then switches the signs of the states in $S'$, i.e.
\[
O\lvert i\rangle_n = \left\{\begin{array}{cl} \lvert i\rangle_n, &i\not\in S \\ - \lvert i\rangle_n, & i\in S.\end{array}\right.
\]
The circuit implementation of the oracle $O$ usually is the most difficult (and computationally expensive) part of the algorithm, and one of the most basic implementations requires the usage of phase kickback~\cite{cleve-QuantumAlgorithmsRevisited-1998}.

The final component of Grover's Algorithm is the diffusion operator $D$, which can be thought of as a reflection around the vector $\lvert 0\rangle^{n}$.
As an operator, we have
\[
D = 2\lvert 0\rangle_n\langle 0\rvert_n  - I.
\]
$D$ is usually implemented using phase kickback in the same fashion as the oracle $O$.

Grover's algorithm requires repeated usage of the operator $G = H^{\otimes n} D H^{\otimes n} O$ which has the net effect of reflecting around $\lvert +\rangle^n$, amplifying the amplitudes of marked states and decreases those of other states.
Measuring the state $G^rH^{\otimes n} \lvert 0\rangle_n$ ($r\ge 1$) gives one of the marked states with high probability, and this probability is maximum when $r = \lfloor\frac{\pi}{4}\sqrt{2^n/|S'|}\rfloor$.
Since $|S'|$ is not known in general, the $r$ is either randomly selected~\cite{durr-hoyer-QuantumMin-1996,boyer-TightQuantumSearchingBounds-1998}, or is estimated using quantum counting algorithms~\cite{mosca-quantumCounting-1998,aaronson-Counting-2020}.

See Figure~\ref{fig:groverExample} for an example of a circuit implementing unrestricted Grover search with $n=3$, $S'=\{\lvert 010\rangle_n, \lvert 011\rangle_n\}$, $S = \{0,1,\ldots, 7\}$.

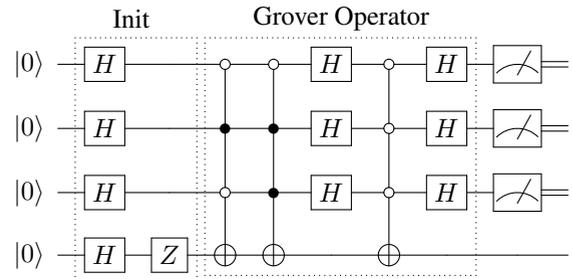
\begin{figure}[ht]
	\[
	\Qcircuit @C=1em @R=1em {
		&\qquad \mbox{Init} & & & &\ \ \ \mbox{Grover Operator} & & & & \\
		\lstick{\ket{0}} & \gate{H} & \qw & \ctrlo{1} & \ctrlo{1}& \gate{H} & \ctrlo{1} & \gate{H}&\meter &\cw \\
		\lstick{\ket{0}} & \gate{H}      & \qw   & \ctrl{1}       & \ctrl{1} & \gate{H} & \ctrlo{1} & \gate{H}&\meter&\cw \\
		\lstick{\ket{0}} & \gate{H}      & \qw        & \ctrlo{1}  & \ctrl{1} & \gate{H} & \ctrlo{1} & \gate{H}&\meter &\cw \\
		\lstick{\ket{0}} & \gate{H}      & \gate{Z}        & \targ       & \targ & \qw & \targ &\qw &\qw & \qw \gategroup{2}{2}{5}{3}{.7em}{.} \gategroup{2}{4}{5}{8}{.7em}{.}
	}
	\]
\caption{An example Grover search implementation\label{fig:groverExample}. In this standard circuit implementation of Grover search, the initialization is achieved by the Hadamard operator $H^{\otimes 3}$. Phase kickback from the fourth qubit initialized to the state $\lvert -\rangle$ is used to negate the amplitudes of $\lvert 010\rangle$ and $\lvert 011\rangle$. Finally, diffusion is achieved via another phase kickback from the same qubit.}
\end{figure}

\subsection{Restricted search space}
Let us now consider a search space $S\subsetneq \{0,1\}^n$. In this case, the algorithm is designed to only evolve over the states of $S$, and this is achieved via an initialization operator $A$ such that
\[
A\lvert 0\rangle_n = \frac 1{\sqrt{|S|}}\sum_{i\in S}\lvert i\rangle_n,
\]
And the Grover operator is changed to a reflection around $A\lvert 0\rangle_n$ instead of $\lvert +\rangle^n$:
\[
G = ADA^{\dagger}O.
\]
The usage of $ADA^{\dagger}$ instead of $H^{\otimes n}DH^{\otimes n}$ makes sure that the evolution of the quantum states remains in the subspace spanned by $S$ instead of the entire space $\{0,1\}^n$, and this leads to probability distribution of the measured outcomes being supported on the state $S$.

The only detail missing in this formulation is the construction of the initialization operator $A$. As we shall see in Section 3, for graph coloring problems (and most applications in general), $A$ can be represented as a block matrix and can be implemented in time linear in the number of qubits.

We make a remark here that the graph coloring algorithm of \cite{saha-CircuitDesignColoringNearTerm-2020} uses the unrestricted formulation of Grover's algorithm, but modifies the oracle to discard states that represent invalid colorings. On the other hand, they do not modify their diffusion operator, leading to states outside the search space having positive probabilities of being measured.


\section{Quantum List Coloring Algorithm}
\label{sec:qlistcoloring}
Our goal in this section is to prove Theorem~\ref{thm:qListColoring}. 
For the remainder of this section, assume that $G=(V,E)$ is an arbitrary graph with $|V|=n$, $|E|=m$.
Further, for every vertex $v$, let $L_v$ denote the list of admissible colors for vertex $v$. 
Then the $i$'th admissible color of the vertex $v$ is denoted as $L_v[i]$.
As our algorithm is based on Grover search, we shall discuss the four basic steps of the algorithm: circuit setup, initialization, oracle and diffusion, each in their respective subsections.

\subsection{Setup and qubit labels}
Our algorithm design requires three different qubit registers:
\begin{itemize}
	\item A \textit{vertex register} to keep track of vertex colors. For each vertex $v$, we require the usage of $\lceil\log_2 |L_v|\rceil$ qubits to represent each color in $L_v$.
	Let us denote $j_v := \lceil\log_2 |L_v|\rceil$, and let the qubits corresponding to vertex $v\in V$ be labeled by $q_{v}^{1},\ldots,q_v^{j_v}$.
	
	\item An \textit{edge register} consisting of $m$ qubits, one corresponding to each edge.
	Let $q'_{uv}$ denote the qubit corresponding to an edge $uv\in E$.
	\item A single qubit \textit{ancilla register}, used for phase kickback in Grover's algorithm. Let $q^\ast$ denote this ancilla.
\end{itemize}
The total number of qubits required is $\sum\limits_{v\in V}j_v + m + 1$.
Now we take a closer look at our list coloring algorithm.

\subsection{Initialization}
For each $v\in V$, we initialize qubits $q_v^1,\ldots,q_v^{j_v}$ to a uniform superposition $\chi_v = \frac1{\sqrt{|L_v|}}\sum\limits_{i=0}^{|L_v|-1}\lvert i\rangle$. This can be achieved via a unitary operator $U_v$ such that
\[
U_v\lvert 0\rangle = \frac1{\sqrt{|L_v|}}\sum_{i=0}^{|L_v|-1}\lvert i\rangle.
\]
We show one way of constructing the operator $U_v$. First, consider the standard basis $B = \{\lvert i\rangle: i\in \{0,1,\ldots,2^{v_j}-1\}\}$. We shall replace any one entry $\lvert i\rangle$ with $\chi_v$, where $i\in L_v$: let $B' = \{\chi_v\}\bigcup (B\setminus \{\lvert i\rangle\})$. It can be seen that $\text{Span}(B)=\text{Span}(B')$. We can now consider $B'$ as an ordered basis with its first entry as $\chi_v$, and apply the Gram-Schmidt process to turn $B'$ into an orthonormal basis $B''$~\cite{pursell-gramSchmidtGaussElim-1991}. Note that as $\|\chi_v\|=1$, it remains unchanged in $B''$. The transpose of the coefficients of the vectors in $B''$ constitutes a change of basis operator that maps $\lvert 0\rangle$ to $\chi_v$, and this is how one can construct the matrix for $U_v$.
It can then be implemented with quantum gates using the results of \cite{li-UnitaryDecomposition-2013,mottonen-MultiqubitGates-2004,vartiainen-EfficientQuantumGates-2004} for example.

For the sake of clarity of the above procedure, let us consider an example with $j_v=2$ and $L_v=\{0,1,2\}$. Note that $\chi_v=\frac1{\sqrt3}(\lvert 0\rangle + \lvert 1\rangle + \lvert 2\rangle)$. Then, $B=\{\lvert 0\rangle, \lvert 1\rangle, \lvert 2\rangle, \lvert 3\rangle\}$, and we can take $B' = \{\chi_v, \lvert 1\rangle, \lvert 2\rangle, \lvert 3\rangle\}$. After the Gram-Schmidt process, we obtain
\[
\begin{aligned} B'' = &\left\{\frac1{\sqrt3}\lvert0\rangle\right.+\frac1{\sqrt3}\lvert1\rangle+\frac1{\sqrt3}\lvert2\rangle,\ \\
&-\frac {1}{\sqrt6}\lvert 0\rangle +\frac{\sqrt2}{\sqrt3}\lvert 1\rangle -\frac 1{\sqrt6}\lvert2\rangle,\\
&\left.-\frac1{\sqrt2}\lvert0\rangle+\frac1{\sqrt2}\lvert2\rangle, \lvert 3\rangle \right\}.\end{aligned}
\]
Hence, in this case, we get 
\[U_v = \begin{bmatrix}\frac1{\sqrt3} & \frac 1{\sqrt3} & \frac1{\sqrt3} & 0\\-\frac1{\sqrt6}&\frac{\sqrt2}{\sqrt3}&-\frac1{\sqrt6}& 0\\-\frac1{\sqrt2}&0&\frac1{\sqrt2}&0\\0&0&0&1\end{bmatrix}^\top.\]
Observe that by construction, $U_v\lvert 0\rangle = U_v\begin{bmatrix}1 & 0 & 0 & 0\end{bmatrix}^\top = \begin{bmatrix} \frac1{\sqrt3} & \frac 1{\sqrt3} & \frac1{\sqrt3} & 0\end{bmatrix}^\top$, as desired.

Finally, let us denote $A = \bigotimes\limits_{v\in V}U_v$.
Then, 
\begin{equation}
\label{eq:init-A}
A \otimes \bigotimes\limits_{e\in E} I\otimes ZH
\end{equation}
is the full initialization operator applied to the circuit starting from 
$\lvert 0\rangle_{\sum\limits_{v\in V}j_v}\otimes \lvert 0\rangle_m\otimes\lvert0\rangle$.
This creates the quantum state $\bigotimes\limits_{v\in V}\chi_v \otimes \lvert 0\rangle_m \otimes \lvert -\rangle$.

\subsection{Oracle}
Traditionally for the graph coloring problem, each vertex color is represented using the same number of qubits, and binary comparator circuits~\cite{wang-TernaryGrover-2011} are used to make sure that the two colors corresponding to two adjacent vertices are different.

While this approach is very efficient for the $k$-coloring problem where every vertex has the same set of admissible colors, the list coloring problem may sometimes require a large number of qubits. In fact, the total number of qubits required for implementing a comparator-based oracle would be $n\cdot\max_{v\in V}\log_2\lceil\max L_v\rceil + m + 1$, which can be much higher than our proposed oracle when the $j_v$'s are not all equal.

In short, for every edge $uv\in E$, we shall encode all possible colorings in $L_u\times L_v$ via flipping the amplitudes of the states corresponding to valid colorings 
\[S = \left\{\lvert i_1\rangle_{j_u}\lvert i_2\rangle_{j_v}: \begin{array}{cc}0\le i_1< |L_u|, 0\le i_2<|L_v|,\\ L_u[i_1]\neq L_v[i_2]\end{array}\right\}.\] 
We propose a classical $O(|L_u|^2|L_v|^2)$-time algorithm to construct an efficient oracle $O_{u,v}$ for flipping these amplitudes. In short, $O_{u,v}$ should have the following net effect:
\begin{equation}
\label{eq:oracleReq}
\begin{aligned}
&O_{u,v}\left(\frac1{\sqrt{|L_u||L_v|}}\sum_{\scriptsize\begin{array}{c} 0\le i_1< |L_u|\\ 0\le i_2< |L_v|\end{array}}\lvert i_1\rangle \lvert i_2\rangle\lvert -\rangle\right) \\&=\frac1{\sqrt{|L_u||L_v|}} \left(\sum_{\lvert i_1\rangle\lvert i_2\rangle \in S}-\lvert i_1\rangle \lvert i_2\rangle + \sum_{\lvert i_1\rangle\lvert i_2\rangle \not\in S}\lvert i_1\rangle\lvert i_2\rangle\right)\lvert -\rangle
\end{aligned}
\end{equation}

Given a string $s$ of length $\ell$ and a subset $J=\{j_1,j_2,\ldots,j_r\}\subseteq \{1,2,\ldots,\ell\}$ we use $s_J$ to denote the substring $s_{j_1}s_{j_2}\ldots s_{j_r}$. $\bbF_2$ denotes the finite field of two elements. Our algorithm for implementing $O_{u,v}$ (Algorithm ~\ref{algo:oracle-uv}) makes use of a subroutine called \ref{func:oracleReduction} that can significantly simplify the complexity and the number of controlled not operations required in many cases.

\begin{algorithm}[ht]
	\SetKwInOut{Input}{Input}\SetKwInOut{Output}{Output}
	\Input{Sets $L_u$, $L_v$ (denote $j_u = \lceil \log_2|L_u|\rceil$, $j_v = \lceil \log_2|L_v|\rceil$).}
	\Output{A $j_u+j_v+1$-qubit quantum circuit $O_{u,v}$ satisfying (\ref{eq:oracleReq}).}
	\caption{Oracle $O_{u,v}$ for marking valid colorings.\label{algo:oracle-uv}}
	\BlankLine
	Let $W'=\text{oracleReduction}(L_u, L_v)$\;
	Create a circuit $C$ with quantum wires $q_1,\ldots, q_{j_u+j_v+1}$\;
	\For{every pair $(J,s)$ in $W'$}{
		Add a multicontrolled NOT gate to $C$ with controls on wires $\{q_j: j\in J, s_j=1\}$, anticontrols on wires $\{q_j: j\in J, s_j=0\}$ and target $q_{j_u+j_v+1}$.
	}
	\Return{$C$}
\end{algorithm}

\begin{function}[ht]
	\SetKwInOut{Input}{Input}\SetKwInOut{Output}{Output}
	\Input{Sets $L_u$, $L_v$.}
	\Output{A set of pairs $W'$.}
	\caption{oracleReduction($L_u$, $L_v$)\label{func:oracleReduction}}
	\BlankLine
	Denote $j_u = \lceil \log_2|L_u|\rceil$, $j_v = \lceil \log_2|L_v|\rceil$\;
	Let $A=\{0,\ldots,|L_u|-1\}$ and $B=\{0,\ldots,|L_v|-1\}$\;
	Convert each element of $A$ and $B$ into $\{0,1\}$-strings of lengths $j_u$ and $j_v$, respectively\;
	Let $X=\{ab: a\in A, b\in B, L_u[a]\neq L_v[b]\}$, where ``$ab$" denotes concatenation\;
	Let $Y=\{ab: a\in A, b\in B\}$, then $|Y|=|A||B|=|L_u||L_v|$\;
	Set $W=\varnothing$\;
	\For{$k=1$ \KwTo $j_u+j_v$}{
		\For{every $k$-element subset $J$ of $\{1,2,\ldots,j_u+j_v\}$}{
			\For{every $\{0,1\}$-string $s$ of length $k$}{
				Add a $\{0,1\}$-variable $x_J^s$ to $W$\;
			}
		}
		Create an empty linear system of equations $\mathcal L$ over $\bbF_2$ with variables $W$\;
		\For{every $\{0,1\}$-string $t\in Y$}{
			Calculate the expression $f(W,t) = \sum \{x_J^s\in W: t_{J} = s\}$\;
			\If{$t\in X$}{
				Add $f(W,t)=1$ to $\mathcal L$\;
			}
			\Else{
				Add $f(W,t)=0$ to $\mathcal L$\;
			}
		}
		Solve the $|L_u||L_v|\times \sum_{j=1}^k\binom{j_u+j_v}{j}\cdot 2^j$ system $\mathcal L$ using Gaussian Elimination over $\bbF_2$\;
		\If{the system $\mathcal L$ is solvable}{
			Solve the linear program minimizing $\sum_{x_J^s\in W} |J|\cdot x_J^s$ subject to $\mathcal L$\;\label{algostep:linearprogram}
			Let $W'=\{(J,s) : x_J^s\in W, x_J^s = 1\}$\;
			\Return{$W'$}
		}
	}
\end{function}

We shall now demonstrate the correctness of Algorithm~\ref{algo:oracle-uv}.

\begin{thm}
	Algorithm~\ref{algo:oracle-uv} gives a circuit $C$ satisfying (\ref{eq:oracleReq}). Further, assuming that the cost of implementing a $k$-controlled NOT operation is k (refer to Remark~\ref{rmk:cost-kCNOT}), the cost of $C$ is the smallest among all circuits that can be made using only controlled NOT operations onto the phase flip qubit.
\end{thm}
\begin{proof}
	It is sufficient to verify the action of $C$ on the states $\lvert i_1\rangle \lvert i_2\rangle \lvert -\rangle$, where $i_1< |L_u|$ and $i_2< |L_v|$.
	Notice that a single controlled NOT gate corresponding to a pair $(J,s)$ in $C$ effectively flips the amplitudes of all basis states represented by $\{0,1\}$-strings $x$ of length $j_u+j_v$ for which $x_i=s_i$, $i\in J$.
	In other words, \emph{all} states of the following form are flipped (here $\ast$ denotes a wildcard, and $J=\{j_1,\ldots, j_k\}$, $s=s_1\cdots s_k$): 
	\[
	\lvert \ast\cdots\ast \underset{j_1\text{'th}}{s_1} \ast\cdots \ast \underset{j_2\text{'th}}{s_2}\cdots \underset{\cdots}{\cdots}\cdots \underset{j_k\text{'th}}{s_{k}}\ast \cdots \ast\rangle
	\]
	Thus, after application of all the controlled NOT gates, only states which appeared in an odd number of $(J,s)$-pairs in $W'$ will survive, and those appearing an even number of times will not have their amplitudes flipped.
	
	Let us now fix a string $t = t_1t_2\cdots t_{j_u+j_v}$, where $t_1\cdots t_{j_u}\in L_u$ and $t_{j_u+1}\cdots t_{j_u+j_v}\in L_v$, and analyze the function~\ref{func:oracleReduction} closely. Recall that $A=\{0,\ldots,|L_u|-1\}$, $B=\{0,\ldots,|L_v|-1\}$, $Y=\{ab: a\in A, b\in B\}$ and $X=\{ab:a\in A, b\in B, L_u[a]\neq L_v[b]\}$.
	We now make a crucial observation: the number of times the amplitude of $\lvert t\rangle$ gets flipped by $C$ is
	\[
	\begin{aligned}
	\left|\left\{	(J,s)\in W': t_J=s	\right\}\right| &= \left|\left\{	x_J^s\in W: x_J^s=1, t_J=s	\right\}\right| \\&= f(W,t),
	\end{aligned}
    \]
	as $x_J^s$ are all $\{0,1\}$-valued. Since the linear system $\mathcal L$ over $\bbF_2$ exactly contains the equations $f(W,t)=1$ if $t\in X$ and $f(W,t)=0$ if $t\in Y\setminus X$, any solution to the system $\mathcal L$ will give a correct circuit satisfying (\ref{eq:oracleReq}).
	\hfill{$\blacksquare$}
	
	Now we come to the second assertion of the theorem. Note that for each variable $x_J^s\in W$ with $x_J^s=1$, we require a $|J|$-controlled NOT operation onto the phase flip ancilla. This means that the cost of $C$ is exactly $\sum_{x_J^s\in W} {|J|}\cdot x_J^s$. As we have minimized this cost via a linear program in Step~\ref{algostep:linearprogram} of \ref{func:oracleReduction}, this proves the second claim.
	
\end{proof}

\begin{rmk}
\label{rmk:cost-kCNOT}
Suppose that the actual cost of implementing $k$-controlled operations is $f(k)$. In the construction of \ref{func:oracleReduction}, we could use any $f(k)$ by modifying the objective function in the linear program of Step~\ref{algostep:linearprogram} to $\sum_{x_J^s\in W} f(|J|)\cdot x_J^s$. 
For an increasing function $f(k)$, the optimization problem discourages usage of gates with high number of controls.
It can be seen in Barenco et. al.~\cite{barenco1995elementary}, Corollary 7.4, for example, that $f(k)\le 48k-204$ for $k\ge 7$.
This implies $f(k) = O(k)$, hence our choice of $f(k)=k$ since the presence of constants doesn't change the optimization problem in Step~\ref{algostep:linearprogram}.
Finding the best $f(k)$ is still an active area of research (see, for e.g., \cite{nakanishi2021quantum,krol2021efficient}).
\end{rmk}

\bigskip

We now have all the ingredients required to implement our full oracle, which is presented as Algorithm~\ref{algo:oraclefull} below.
\begin{algorithm}[ht]
\SetKwInOut{Input}{Input}\SetKwInOut{Output}{Output}
\Input{Graph $G=(V,E)$, lists $\{L_v:v\in V\}$, and a $\sum_{v\in V}j_v + m +1$-qubit quantum circuit $C$.}
\Output{$C$ appended with a graph coloring oracle $O$.}
\caption{Full oracle $O$ for marking valid list colorings\label{algo:oraclefull}} 
\BlankLine
\For{every edge $uv\in E$}{
	Apply Oracle $O_{u,v}$ on qubits $q_u^1,\ldots, q_u^{j_u}$, $q_v^1,\ldots, q_v^{j_v}$, $q'_{uv}$ of $C$\;
}
Apply a controlled NOT operation with controls $q'_{\ast}$ and target $q^\ast$\;
\For{every edge $uv\in E$}{
	Apply Oracle $O_{u,v}$ on qubits $q_u^1,\ldots, q_u^{j_u}$, $q_v^1,\ldots, q_v^{j_v}$, $q'_{uv}$ to de-entangle the vertex register of $C$\;
}
\Return{$C$}
\end{algorithm}
\subsection{Diffusion}

Our diffusion operator is very straightforward, and is a direct application of the restricted search space diffusion mentioned in section 2.2. Let $A$ be the initialization operator we implemented in (\ref{eq:init-A}), then a diffusion is achieved by $A D A^\dagger$, where $D=2\lvert 0\rangle\langle 0\rvert - I$ can be implemented by a controlled NOT with anticontrols on each of the vertex qubits and target the phase flip ancilla.

Refer to Figure~\ref{fig:fullalgo} for an illustration of the full list coloring algorithm.
\begin{figure}[h]
	\includegraphics[width=0.5\textwidth]{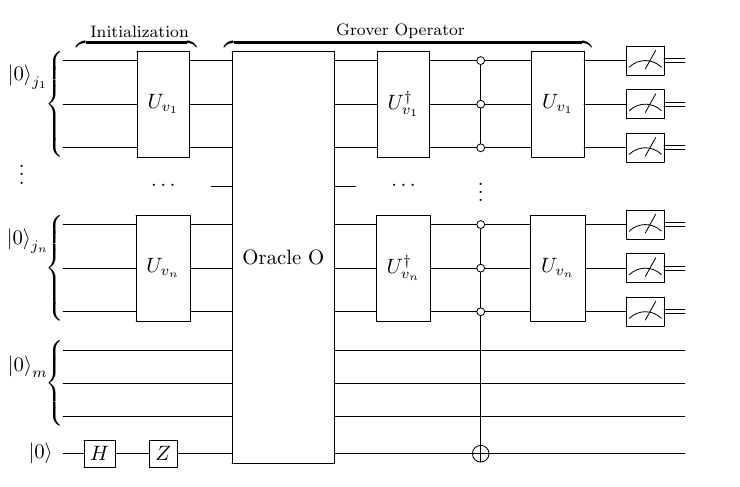}
	\caption{Outline of our list coloring circuit for a single Grover iteration.\label{fig:fullalgo}}
\end{figure}

\section{Results}
\label{sec:results}
We implement our list coloring algorithm in \texttt{python 3.8} using \texttt{blueqat-sdk}.
In order to gauge the efficiency of our result, we run experiments for the $3$ and $4$-coloring problems used in~\cite{saha-CircuitDesignColoringNearTerm-2020} on the the Amazon Statevector simulator.

\subsection{$3$-coloring $K_3$}
$3$-coloring the triangle graph $G=(\{1,2,3\},\{12,23,13\})$ is equivalent to the list-coloring problem on $G$ with $L_1=L_2=L_3=\{\lvert00\rangle, \lvert01\rangle, \lvert10\rangle\}$. In this case, our oracle component $O_{uv}$ obtained from Algorithm~\ref{algo:oracle-uv} is:
	\[
O_{uv}  =   \Qcircuit @C=.7em @R=.7em {
	\lstick{} & \qw & \ctrlo{3} & \qw & \qw \\
	\lstick{} & \qw & \qw & \ctrlo{1} & \qw \\
	\lstick{} & \qw & \qw & \ctrlo{2} & \qw \\
	\lstick{} &\qw  & \ctrlo{1} & \qw & \qw \\
	\lstick{} & \qw & \targ  & \targ & \qw \\
}
    \]
The initialization operator $U_v$ can be written as the following circuit:

\[
U_v  =   \Qcircuit @C=.4em @R=.7em {
	\lstick{} & \gate{R_y(2\arcsin\sqrt{2/3})} & \qw & \qw & \qw & \ctrl{1} & \qw & \qw & \gate{X} & \qw \\
	\lstick{} &\qw &\gate{S} & \gate{H} & \gate{T} & \targ & \gate{T^\dagger} & \gate{H} & \gate{S^\dagger} & \qw \\
}\]
Each Toffoli gate can be decomposed into two-qubit gates using the standard decomposition:
\[
\Qcircuit @C=.5em @R=.7em {
	\lstick{} & \qw      & \qw     &\qw             &\ctrl{2}&\qw     &\qw      & \qw    &\ctrl{2}   & \qw             &\ctrl{1} & \qw & \ctrl{1} & \gate{T} & \qw \\
	\lstick{} & \qw      & \ctrl{1}&\qw             &\qw     &\qw     &\ctrl{1} & \qw         &\qw   & \gate{T^\dagger}&\targ & \gate{T^\dagger} & \targ & \gate{S} & \qw\\
	\lstick{} & \gate{H} & \targ   &\gate{T^\dagger}&\targ   &\gate{T}&\targ &\gate{T^\dagger}&\targ & \gate{T}        &\gate{H}& \qw & \qw & \qw  & \qw\\
}
\]
To implement a \texttt{cccx} gate, we use a clean ancilla qubit to reduce circuit depth as demonstrated below:
\[
\begin{array}{ccc}
\Qcircuit @C=.7em @R=1.2em {
	\lstick{} & \ctrl{1} & \qw\\
	\lstick{} & \ctrl{1} & \qw\\
	\lstick{} & \ctrl{1} & \qw\\
	\lstick{} & \targ    & \qw\\
}\qquad &\stackrel{\qquad}{=} \qquad
&\Qcircuit @C=.7em @R=1.2em {
	\lstick{}            & \ctrl{1} & \qw      &\ctrl{1}&\qw\\
	\lstick{}            & \ctrl{3} & \qw      &\ctrl{3}&\qw\\
	\lstick{}            & \qw      & \ctrl{1} & \qw    &\qw\\
	\lstick{}            & \qw      & \targ    & \qw    &\qw\\
	\lstick{\lvert 0\rangle} &\targ & \ctrl{-1}& \targ  &\qw
}
\end{array}
\]


Finally, we are able to run our circuit on the Amazon Statevector Simulator after decomposing into these elementary single and two-qubit operations. Our resulting histogram for \emph{one Grover iteration} is shown in Figure~\ref{fig:simulatorK3}.
\begin{figure}[ht]
    \includegraphics[width=.5\textwidth]{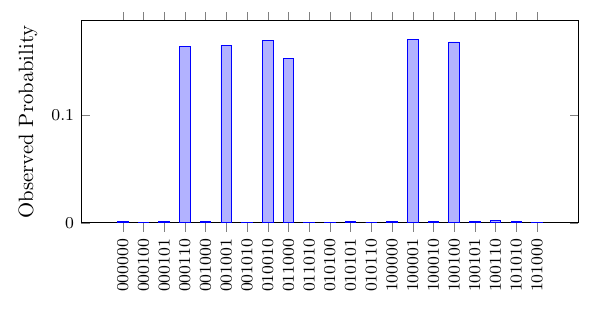}
	\caption{\label{fig:simulatorK3}Observed frequencies for 2000 shots for 3-coloring $K_3$.}
\end{figure}

It is seen that states in the set $A = \{\lvert 000110\rangle,$ $\lvert 001001\rangle,$ $\lvert 011000\rangle,$ $\lvert 011000\rangle,$ $\lvert 100001\rangle,$ $\lvert 100100\rangle\}$ are all massively amplified. Each state of $A$ represents a valid $3$-coloring of $K_3$. Let $B=\{\lvert 00\rangle,\lvert 01\rangle, \lvert 10\rangle\}^{\otimes3}$.
The statevector at the end of our algorithm can be calculated theoretically, and we present the probability density function over the 6-qubit states in Table~\ref{tab:theoreticalProbK3}.
In particular, each state in $A$ has a probability $p_A$ of $0.165066$ to be measured, and states containing $\lvert 11\rangle$ as a color have probability $0$ of being observed. Our algorithm has a single-shot accuracy of $p_A\cdot |A| \approx 0.9904$.

\begin{table}[h]
    \centering
    \begin{tabular}{|c|c|c|}
    \hline
    Set of states & Representation & PDF value\\
    \hline \hline
    Marked states & $A$ & 0.165066 \\
    \hline
    Unmarked but valid combinations & $B\setminus A$ & 0.000457\\
    \hline
    Invalid combinations & $\{\lvert 0\rangle, \lvert 1\rangle\}^{\otimes 6} \setminus B$ & 0\\
    \hline
    \end{tabular}
    \caption{Theoretical probability density for 3-coloring $K_3$.}
    \label{tab:theoreticalProbK3}
\end{table}
 
\subsection{$4$-coloring $K_4$}
Let $G=(\{1,2,3,4\},\{12,13,14,23,24,34,14\})$ be the complete graph on $4$ vertices, and suppose $L_1=L_2=L_3=L_4 = \{\lvert 00\rangle, \lvert 01\rangle, \lvert 10\rangle, \lvert 11\rangle\}$.
In this case the simple Hadamard operator $H^{\otimes 2}$ initializes each vertex register to a uniform superposition of its valid colors, and we can then run Algorithm~\ref{algo:oracle-uv} to figure out the component $O_{uv}$. It turns out that one of the valid solutions minimizing the cost of gates used is the following circuit:
\[
\begin{array}{ccc}
O_{uv} & = &  \Qcircuit @C=.7em @R=.7em {
	\lstick{} & \qw & \ctrl{3} & \qw & \ctrlo{1} & \qw     & \qw \\
	\lstick{} & \qw & \qw & \ctrl{1} & \ctrlo{3} & \qw     & \qw \\
	\lstick{} & \qw & \qw & \ctrl{2} & \qw       &\ctrlo{1}& \qw \\
	\lstick{} &\qw  & \ctrl{1} & \qw & \qw       &\ctrlo{1}& \qw \\
	\lstick{} & \qw & \targ  & \targ & \targ     &\targ    & \qw \\
}
\end{array}.
\]
Figure~\ref{fig:statevecK4} shows the results of running our circuit with \emph{one Grover iteration} on the statevector simulator.
\begin{figure}[H]
\centering
\includegraphics[width=0.5\textwidth]{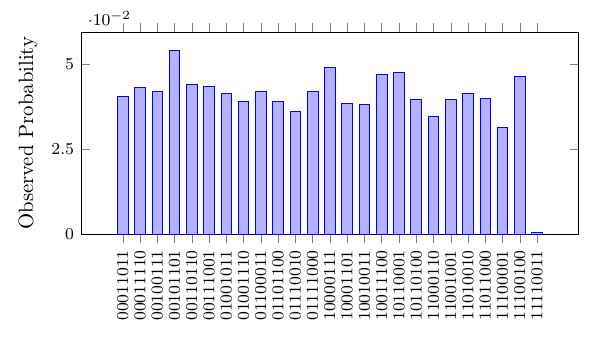}
\caption{\label{fig:statevecK4}Observed frequencies for 2000 shots for 4-coloring $K_4$.}
\end{figure}

In this case, the result of a statevector computation after one Grover iteration is shown in Table~\ref{tab:theoreticalProbK4}.
Here $A'$ denotes the set of all $4!=24$ marked states corresponding to valid colorings of $K_4$, and $B'=\{\lvert 0\rangle, \lvert 1\rangle\}^{\otimes 8}$ is the full state space.

\begin{table}[h]
    \centering
    \begin{tabular}{|c|c|c|}
    \hline
    Set of states & Representation & PDF value\\
    \hline \hline
    Marked states & $A'$ & 0.041657 \\
    \hline
    Unmarked states & $B'\setminus A'$ & $10^{-6}$\\
    \hline
    \end{tabular}
    \caption{Theoretical probability density for 4-coloring $K_4$.}
    \label{tab:theoreticalProbK4}
\end{table}

The result of our experiment can also be seen to follow this distribution.

\subsection{Comparison with Previous Work: $3$-coloring $K_3$}
\begin{figure}[h]
    \centering
    \includegraphics[width=0.5\textwidth]{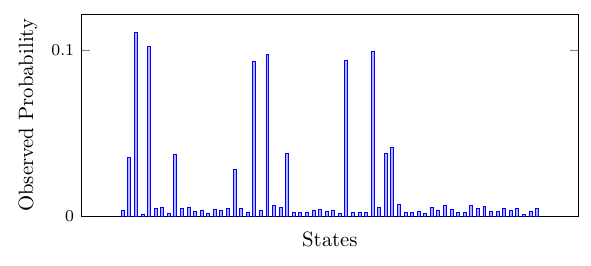}
    \caption{Observed probabilities for 2000 shots for $3$-coloring $K_3$ using the algorithm of \cite{saha-CircuitDesignColoringNearTerm-2020}.}
    \label{fig:sahaResultsK3}
\end{figure}

We take \cite{saha-CircuitDesignColoringNearTerm-2020} as the state-of-the-art result for the $k$-coloring problem.
We compare the performance of our algorithm with theirs for $k=3$.
Figure~\ref{fig:sahaResultsK3} depicts the empirical probability distribution obtained via implementing their algorithm for one Grover rotation and running on the Amazon Statevector simulator.
We observe that our algorithm has much higher probability of selecting a correct state (around $0.9904$), whereas their method has an empirical success rate of $0.5965$ for $2000$ shots.

Finally, we compare the gate count of their circuit versus ours for the 3-coloring $K_3$ problem for one Grover iteration in Table~\ref{tab:gateCountsK3}.

\begin{table}[h]
    \centering
    \begin{tabular}{|c|c|c|}
    \hline
    Gates & Saha et. al.\cite{saha-CircuitDesignColoringNearTerm-2020} & Our Circuit\\
    \hline \hline
    $X$, $H$ & 58 & 66  \\ \hline
    $R_y$, $R_z$ & 0 & 45 \\ \hline
    \texttt{cx} & 20 & 15 \\ \hline
    \texttt{ccx} & 24 & 12 \\ \hline
    \texttt{cccx} & 3 & 2 \\ \hline
    \texttt{ccccx} & 1 & 0 \\ \hline \hline
    Total elementary gates & 686 & 429 \\ \hline
    \end{tabular}
    \caption{Comparison of gate counts for one Grover iteration.}
    \label{tab:gateCountsK3}
\end{table}

The code we used for this comparison is available at the following GitHub repository: \url{https://github.com/Potla1995/Grover-ListColoring}.

\subsection{Experiments on NISQ devices and limitations}
We also ran our circuits on the IonQ physical device. However, reasonable results were not obtained. There might be several possible reasons behind this, such as:
\begin{enumerate}
    \item \texttt{cx} is not a gate natively implemented on the IonQ processor. Although its fully connected topology is ideal for the circuits we construct, controlled operations need to be expressed in terms of M\o{}lmer-S\o{}rensen gates~\cite{molmerSorensenIonQGates1999}, which makes them expensive.
    \item Our circuit has a depth of $30$ when expressed in terms of the elementary gates (including \texttt{cx}). This makes the circuit much more susceptible to errors from noise and decoherence. Although we can implement and run our algorithm on quantum devices currently available, the results suggest that devices with lesser intrinsic noise are needed for practical use of our algorithm.
\end{enumerate}
These limitations suggest that while our algorithm is designed for running on NISQ devices, the circuits generated by it are too complex, giving unusable results due to the intrinsic noise in the currently available quantum devices.
In any case, classical computers can brute-force the list coloring problem for the small-sized graphs that can be currently encoded on an NISQ device using our encoding scheme.


\section{Applications and Concluding Remarks}
\label{sec:conclusion}

The list coloring problem is ubiquitous in real life, as it not only generalizes an already well-appearing problem of graph coloring (scheduling, satisfiability etc.), but is also applicable to several other scenarios, such as:
\begin{itemize}
	\item \textbf{Wireless network Allocation}~\cite{wang-ListColoringWirelessAllocation-2005}:  In a wireless network, each radio is allocated special frequencies which it can connect to. Suppose that radios in close proximity cannot operate on the same frequency due to interference. The problem of which radio is connected to which network frequency can be modeled as a list-coloring problem in the following graph: let the radios be represented by vertices, and add an edge between two radios if they are in close proximity of each other. The lists for each vertex will be the set of available frequencies for its corresponding radio.
	\item \textbf{Register Allocation}: In compiler optimization, register allocation is the process of assigning a large number of target program variables ($n$) onto a small number of CPU registers ($k$), which reduces to a $k$-coloring problem on an $n$-vertex graph. 
	\item \textbf{Sudoku}: We can represent every cell in a sudoku problem with a vertex, and join two vertices with an edge if they are in same row or same column or same block. Given $x$ already filled cells, we can formulate the sudoku problem as a list-coloring problem on $81-x$ vertices and at most $9$ colors.
\end{itemize}

We proposed a Grover search-based quantum algorithm that achieves quadratic speedup in query complexity compared to a classical brute force search, and also proposed a classical algorithm that can simplify the oracle design for several special instances of the list coloring problem. We demonstrate the efficiency of our method in comparison with previous work by running our algorithm on the Amazon statevector simulator for the $3$ and $4$-coloring problems.

Unfortunately, the list coloring problem is difficult to solve both classically and using quantum algorithms, as for generic lists with no known structure, brute force seems to be the only way to attack the problem. As our algorithm is basically a brute force quantum search with some optimizations in the oracle, it is expected to perform better in general cases where the structure of the lists are unknown. However, the existence of clever hybrid algorithms exploiting specific structures for known lists cannot be ignored, and is a very promising future direction.

\bigskip
Finally, we note that one can obtain an improvement on our algorithm by just changing a given list coloring problem to a reduced problem. For example, if $G=(V,E)$ has a vertex $v$ with $|L_v|=1$, we can color $v$ first and remove its color from each $L_u$ such that $uv\in E$, and iterate until all lists have at least two colors.

\section*{Acknowledgments}
Part of this work was completed during a summer internship at Elyah Co. in 2020. The author is grateful to Sydney Andrews, Ry\=utar\=o Nagai and Goutham Rajendran for several helpful discussions and comments; Y\=uichir\=o Minato and Salman Al Jimeely for encouragement and support; and the anonymous reviewers whose insightful comments and suggestions helped drastically improve the quality of the manuscript.

\bibliographystyle{plain}
\bibliography{qListColoring.bbl}
\EOD

\end{document}